\documentclass[smallextended]{svjour3}
\usepackage[english]{babel}
\usepackage{amsmath}
\usepackage{amsfonts}
\usepackage{amssymb}
\usepackage{graphicx}
\def\be{\begin{equation}}
\def\ee{\end{equation}}

\newcommand{\E}{\mathbb{E}}
\newcommand{\R}{\mathbb{R}}

\newcommand{\qp}{q_{L_p}}
\newcommand{\qpp}{q_{L_{p+1}}}

\newcommand{\va}{a}
\newcommand{\vlambda}{\lambda}

\begin{document}

\title{Annealing and replica-symmetry in Deep Boltzmann Machines}
\dedication{Dedicated to Joel Lebowitz with profound admiration}

\author{Diego Alberici, Adriano Barra, Pierluigi Contucci, Emanuele Mingione}
\institute{D. Alberici, P. Contucci,  \at Dipartimento di Matematica, Universit\`a di Bologna, Italy \and
A. Barra \at Dipartimento di Matematica e Fisica Ennio De Giorgi, Universit\`a del Salento, Italy \& Istituto Nazionale di Fisica Nucleare, Sezione di Lecce, Italy \and
E. Mingione \at Dipartimento di Matematica, Universit\`a di Bologna, Italy \& EPFL, Switzerland}

\maketitle

\begin{abstract}
In this paper we study the properties of the quenched pressure of a multi-layer spin-glass model (a deep Boltzmann Machine in artificial intelligence jargon) whose pairwise interactions are allowed between spins lying in adjacent layers and not inside the same layer nor among layers at distance larger than one. We prove a theorem that bounds the quenched pressure of such a K-layer machine in terms of K Sherrington-Kirkpatrick spin glasses and use it to investigate its annealed region. The replica-symmetric approximation of the quenched pressure is identified and its relation to the annealed one is considered.\\
The paper also presents some observation on the model's architectural structure related to machine learning. Since escaping the annealed region is mandatory for a meaningful training, by squeezing such region we obtain thermodynamical constraints on the form factors. Remarkably, its optimal escape is achieved by requiring the last layer to scale sub-linearly in the network size.
\end{abstract}

\keywords{Spin glasses, Boltzmann machines, machine learning, thermodynamical constraints.}

\section{Introduction and results}
The rigorous approach to the study of the spin glass phase started with the celebrated result by Aizenman, Lebowitz and Ruelle \cite{ALR} of the annealed regime for the Sherrington and Kirkpatrick (SK) model more than three decades ago. Using a cluster expansion technique it was proved that the quenched free energy, the one describing the peculiar structure of the spin glass, and the annealed one coincide in the thermodynamic limit when the inverse temperature $\beta$ is smaller than one (high temperature regime). That paper, as a side result, proved also that on such regime the thermodynamic limit exists as a consequence of the simple annealed computation of the free energy.

In \cite{BCMT} a generalisation of the SK model was proposed and studied. The classical permutation group symmetry among the spin particles, a central feature of the mean field formulation of the spin glass model, was replaced by a weaker condition where the symmetry holds only between and within subsets of them. The total number of spins $N$ is therefore split into $K$ homogeneous sets each containing $N_1$, $N_2$, ..., $N_K$ particles with the constraints $\sum_{p=1}^{K}N_p=N$ and
$N_p/N \rightarrow \lambda_p$. For this model it was proved, using a suitable interpolation scheme {\it a la} Guerra \cite{Guerra}, that a Parisi-like bound holds for the free energy density under suitable conditions on the interactions among the spins. The conditions means, essentially, that the strength of the interaction within each homogenous set of particles has to dominate the one between different sets, i.e. the interactions have to be {\it elliptic} which implies the positivity argument that Guerra's interpolation comes endowed with (see also \cite{AMT2018}). The bound was later proved to be sharp, under the same conditions, in a beautiful paper by Panchenko \cite{PanchenkoMSK}. When instead the interaction coefficients are on the {\it hyperbolic} regime the model is beyond the classical techniques available to solve it. The case  $K=2$ is known in the litterature as bipartite spin glass model and has been studied in  \cite{AC,bipartiti}.

In this paper we focus on a special and interesting case of the latter, a {\it deep} SK model, or deep Boltzmann machine. The interest and the name come from the structure of the networks used in a class of machine learning techniques.
Namely, in deep Boltzmann machines, the couplings among neurons (the spins in the physical jargon) are symmetric and this ensures the {\em detailed balance} property: the long term relaxation of any (not-pathological) stochastic neural dynamics converges to the Gibbs distribution of a related cost-function (the Hamiltonian)\cite{Coolen}. All that has a twofold advantage: the first in machine learning, i.e. the possibility to derive explicit learning rule, as e.g. the celebrated {\em contrastive divergence} when extremizing the Kullback-Leibler cross-entropy; the second, in machine retrieval, is that we can import a set of mathematical techniques and ideas originally developed to treat the statistical mechanics of the spin glasses.\cite{MezardMontanari,BovierBook,CG,Nishimori,Amit}.

The paper is organised as follows. In section \ref{definizioni} we introduce the notations and the definitions. In section \ref{lb}, using the aforementioned techniques, we prove that the thermodynamic pressure for the considered class of models is always larger than a suitable convex combination of SK pressures each living on the $p$-th layer.
In section \ref{sec:ann} by using the theorem from the previous section we find a set of parameters where the quenched and annealed pressure coincide, identifying therefore a sufficient condition for the annealed phase to hold and, as a side result, a region where the thermodynamic limit exists. Since such region depends both on the temperature and on the factors lambdas, the section ends with an extremal condition on that region to make it as narrow as possible: satisfying this request is mandatory in machine learning since escaping the annealed region is a paramount necessity to accomplish learning as well as retrieval.

In section \ref{Sec:RS} we identify the replica symmetric solution, i.e. the pressure of the model under the self-averaging condition for the overlap. We study moreover, in the case of zero external field, the solution of the stationary condition for the replica symmetric functional around the origin. By investigating its stability in the cases up to $K=4$, we find a set of conditions that coincide with those ensuring the annealed solution.

\begin{figure}[!]
\begin{center}\label{Figura:Uno}
		\includegraphics[width=10cm]{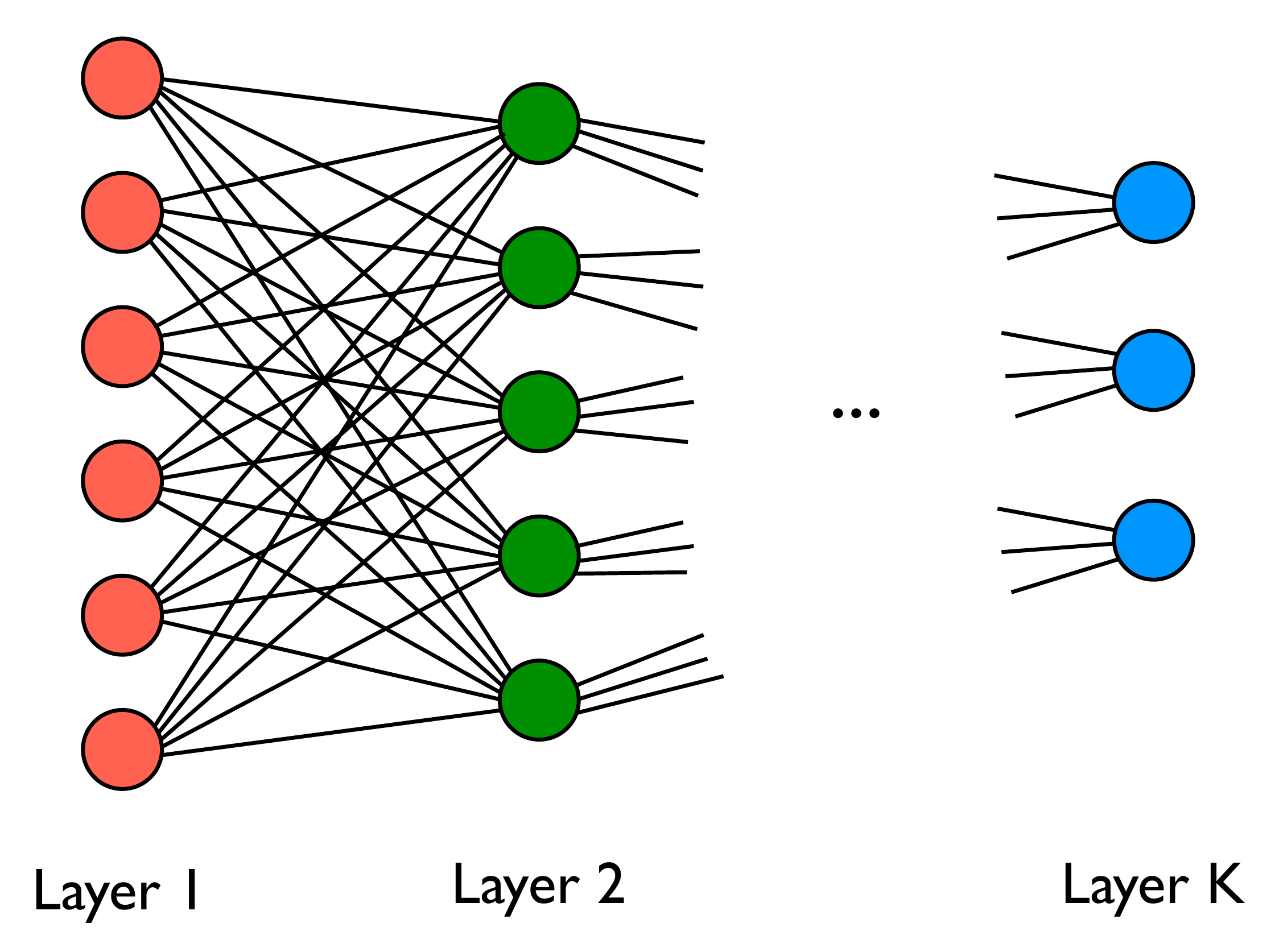}
        \caption{Schematic representation of the deep Boltzmann machine equipped with K layers under study. Each circle represents a binary neuron while all the interactions are drawn among neurons in adjacent layers (but there are no intra-layer interactions).}
\end{center}
\end{figure}

\section{Definitions.}
\label{definizioni}
The Deep Boltzmann Machine [DBM] under investigation here is the original one \cite{Hinton1}:
there are $N$ binary Ising spins
The weights connecting layers $L_p$ and $L_{p+1}$ are $N_p \times N_{p+1}$ real valued i.i.d. random couplings sampled from a Gaussian distribution.
\newline
We assume that the relative sizes, that we refer to as form factors, of the layers converge in the large volume limit:
\be\label{lambdadef}
\lambda_p^{(N)} \equiv\, \frac{N_p}{N} \,\xrightarrow[N\to\infty]{}\, \lambda_p \,\in[0,1]
\ee
for every $p=1,\dots,K\,$.
We denote by $\Lambda_N \equiv (N_1,\dots,N_K)$ the sizes of the layers defining the geometric structure underlying the DBM. Moreover we denote by $\vlambda=(\lambda_1,\dots,\lambda_K)$ the relative sizes in the large volume limit. Observe that $\sum_{p=1}^K\lambda_p=1\,$.
\newline

\begin{definition}
Considering $N$ spins $\sigma=(\sigma_i)_{i=1,\dots,N}\in\{-1,1\}^N$ arranged over $K$ layers $L_1,\dots,L_K$, the Hamiltonian of the (random) Deep Boltzmann Machine [DBM] is
\be \label{eq:H}
H_{\Lambda_N}(\sigma) \,\equiv\, -\frac{\sqrt{2}}{\sqrt{N}}\; \sum_{p=1}^{K-1}\, \sum_{(i,j)\in L_p\times L_{p+1}} J^{(p)}_{ij}\, \sigma_i\sigma_j
\ee
where $J^{(p)}_{ij}$, $(i,j)\in L_p\times L_{p+1}$, $p=1,\dots,K-1\,$ is a family of i.i.d. standard Gaussian random variables coupling spins in the layer $L_p$ to those in the layer $L_{p+1}\,$.\\
\end{definition}

\begin{definition}
Given two spin configurations $\sigma,\tau\in\{-1,1\}^N$, for every $p=1,\ldots,K$ we define the overlap over the layer $L_p$ as
\be
\qp(\sigma,\tau) \,\equiv\, \frac{1}{N_p}\,\sum_{i\in L_p} \sigma_i\,\tau_i \;\in[-1,1] \;.
\ee
\end{definition}

Therefore the covariance matrix of the Gaussian process $H_N$ can be written as
\be\label{eq:cov H}
\E \,H_{\Lambda_N}(\sigma)\, H_{\Lambda_N}(\tau) \,=\, 2\,N\,\sum^{K-1}_{p=1}\,\lambda_p^{(N)}\,\lambda_{p+1}^{(N)}\;\, \qp(\sigma,\tau)\;\qpp(\sigma,\tau)
\ee

\begin{definition}
Given $\beta>0$, the random partition function of the model introduced by the Hamiltonian (\ref{eq:H}) is
\be
Z_{\Lambda_N} (\beta) \,\equiv\, \sum_{\sigma\in\{-1,1\}^N} e^{-\beta\,H_{\Lambda_N}(\sigma)} \;.
\ee
and its quenched pressure density is
\be\label{qpres}
p^{DBM}_{\Lambda_N}(\beta) \,\equiv\, \frac{1}{N}\,\E\,\log Z_{\Lambda_N}(\beta)
\ee
where $\E$ to denote the expectation over all the couplings $J^{(p)}_{ij}\,$.\\
\end{definition}

\begin{remark}\label{addfield}
As it can be useful in machine learning \cite{DLbook}, we may also include a magnetic field within each layer by generalizing the Hamiltonian (\ref{eq:H}) as
\be
H'_{\Lambda_N}(\sigma) \,\equiv\, H_{\Lambda_N}(\sigma) \,+\, \sum_{p=1}^K\,  \sum_{i\in L_P} h^{(p)}_i \sigma_i,
\ee
where for any $p=1,\ldots K$,  $(h^{(p)}_i)_{i\in L_P}$ is family of i.i.d. random variables.
\end{remark}

\section{A lower bound for the quenched pressure of the DBM}\label{lb}

In this section we give an explicit bound for the quenched free energy of the DBM - composed by $K$ layers - in terms of $K$ independent Sherrington-Kirkpatrick spin-glasses [SK] (whose sizes share one-by-one the sizes of the DBM's layers).

Considering $N$ spin variables $\sigma_i$, $i=1,\dots,N$, we recall that the Hamiltonian of the SK model is
\be \label{eq:HSK}
H_{N}^{SK}(\sigma) \,\equiv\, -\frac{1}{\sqrt{N}}\; \sum_{i,j=1}^N J_{ij}\, \sigma_i\sigma_j
\ee
where $J_{ij}$, $i,j=1,\dots,N$ is a family of i.i.d. standard Gaussian random couplings.
Given two spin configurations $\sigma,\tau\in\{-1,1\}^N$, their overlap is
\be
q_N(\sigma,\tau) \,\equiv\, \frac{1}{N}\,\sum_{i=1}^N \sigma_i\,\tau_i \;\in[-1,1]
\ee
and the covariance matrix of the Gaussian process $H_N^{SK}$ is:
\be\label{eq:cov HSK}
\E \,H_{N}^{SK}(\sigma)\, H_{N}^{SK}(\tau) \,=\, N\, q_N^2(\sigma,\tau) \;.
\ee
Given an inverse temperature $\beta>0$, the random partition function of the SK model is
\be
Z_{N}^{SK} (\beta) \,\equiv\, \sum_{\sigma\in\{-1,1\}^N} e^{-\beta\,H_{N}^{SK}(\sigma)}
\ee
and its quenched pressure density is
\be\label{qpresSK}
p^{SK}_{N}(\beta) \,\equiv\, \frac{1}{N}\,\E\,\log Z_{N}^{SK}(\beta)
\ee
where $\E$ to denote the expectation over all the couplings $J_{ij}\,$.
The quenched pressure converges as $N\to\infty$ and we denote its limit by $p^{SK}(\beta)\,$ \cite{Panchenko-Book,GT,Guerra,Tala,MPV}.

Now let $\va = (a_p)_{p=1,\dots,K-1}$ be a sequence of positive numbers.
For every $p=1,\dots, K$ we consider an SK model of size $N_p$ at inverse temperature $\beta\,\sqrt{\lambda_p^{(N)}\,\theta_p(\va)}\,$, where we set
\be\label{gammap}
\begin{cases}
\theta_1(\va) \,\equiv\, a_1 \\
\theta_p(\va) \,\equiv\, \dfrac{1}{a_{p-1}} +\, a_p\ & \textrm{if }p=2,\dots,K-1 \\
\theta_K(\va) \,\equiv\, \dfrac{1}{a_{K-1}}
\end{cases} \ .
\ee
With the notation introduced we have the following:
\newline

\begin{theorem}\label{maint}
The quenched pressure of the DBM described by the cost function \eqref{eq:H} satisfies the following lower bound
\be\label{main}
\begin{split}
p^{DBM}_{\Lambda_N}(\beta) \,\geq\; & \sum_{p=1}^K \lambda_p^{(N)}\; p^{SK}_{N_p}\bigg(\beta\sqrt{\lambda_p^{(N)}\,\theta_p(\va)}\;\bigg) -\, \frac{\beta^2}{2}\,\sum_{p=1}^K \big(\lambda_p^{(N)}\big)^{2}\;\theta_p(\va) \ +\\
& +\, \beta^2\sum_{p=1}^{K-1}\lambda_p^{(N)}\,\lambda_{p+1}^{(N)}
\end{split}
\ee
where $\theta_p(\va)$ is defined by \eqref{gammap} and $\va\in(0,\infty)^{K-1}$ can be arbitrarily chosen.
Therefore:
\be\label{main2}
\begin{split}
\liminf_{N\to\infty}p^{DBM}_{\Lambda_N}(\beta) \,\geq\, &
\sup_{\va\in(0,\infty)^{K-1}} \left\{
\sum_{p=1}^K \lambda_p\; p^{SK}\Big(\beta\sqrt{\lambda_p\,\theta_p(\va)}\;\Big) \,-\,
\frac{\beta^2}{2}\sum_{p=1}^K \lambda_p^2\,\theta_p(\va) \right\} \,+\\
& +\, \beta^2 \sum_{p=1}^{K-1}\lambda_p\lambda_{p+1} \;.
\end{split}
\ee
\end{theorem}

\begin{proof}
For every $p=1,\ldots,K$ let $H_{L_p}^{SK}(s)$, $s\in\{-1,1\}^{L_p}\,$ be a gaussian process representing the Hamiltonian of an SK model over the $N_p$ spin variables in the layer $L_p\,$.
We assume that $H_{L_1}^{SK},\dots,H_{L_K}^{SK}$ are independent processes, also independent of $H_{\Lambda_N}$, the Hamiltonian of the DBM \eqref{eq:H}.
For $\sigma\in\{-1,1\}^N$ and $t\in[0,1]$ we define an interpolating Hamiltonian as follows:
\be\label{inth}
\mathcal H_{N}(\sigma,t) \,\equiv\,
\sqrt{t}\; H_{\Lambda_N}(\sigma) \,+\,
\sqrt{1-t}\;\sum_{p=1}^K\, \sqrt{\lambda_p^{(N)}\,\theta_p(\va)}\,\; H_{L_p}^{SK}(\sigma_{L_p}) \;,
\ee
where of course $\sigma_{L_p}\equiv(\sigma_i)_{i\in L_p}\,$.
An interpolating pressure is naturally defined as
\be\label{intpe}
\varphi_{N}(t)\,\equiv\, \frac{1}{N}\, \E\,\log\,\mathcal Z_{N}(t) \ ,
\ee
where
\be\label{partinn}
\mathcal Z_{N}(t)\,\equiv\, \sum_{\sigma\in\{-1,1\}^N} e^{-\beta\,\mathcal H_{N}(\sigma,t)} \;.
\ee
Observe that the quenched pressure of the DBM and a convex combination of quenched pressures of SK models are recovered at the endpoints of $[0,1]\,$:
\begin{align} \label{t1}
& \varphi_{N}(1) \,=\, p^{DBM}_{\Lambda_N}(\beta) \;,\\
\label{t0}
& \varphi_{N}(0) \,=\, \sum_{p=1}^K \lambda_p^{(N)}\,\; p^{SK}_{N_p}\bigg(\beta\,\sqrt{\lambda_p^{(N)}\,\theta_p(\va)}\;\bigg) \;.
\end{align}

For every function $f:\{-1,1\}^{N}\times\{-1,1\}^{N}\rightarrow \R\,$ we denote
\be\label{qexp}
\left\langle\, f\,\right\rangle_{N,t} \,\equiv\,
\E\,\sum_{\sigma,\tau}\frac{e^{-\beta\,\mathcal H_{N}(\sigma,t)-\beta\,\mathcal H_{N}(\tau,t)}}{\mathcal Z_N^2(t)}\,f(\sigma,\tau) \;.
\ee
Let $Q_N:\{-1,1\}^{N}\times\{-1,1\}^{N}\rightarrow \R\,$,
\be \label{Qn}
Q_N\,\equiv\;
2\sum_{p=1}^{K-1} \lambda_p^{(N)}\lambda_{p+1}^{(N)}\; \qp\,\qpp \,-\,
\sum_{p=1}^K \big(\lambda_p^{(N)}\big)^{2}\, \theta_p(\va)\; \qp^2 \;,
\ee
then Gaussian integration by parts leads to the following result:
\be\label{deriv_ann}
\frac{d\varphi_N}{dt} \,=\,
\frac{\beta^2}{2}\, \Bigg(2\sum_{p=1}^{K-1}\lambda_p^{(N)} \lambda_{p+1}^{(N)} \,-\, \sum_{p=1}^K \big(\lambda_p^{(N)}\big)^{2}\, \theta_p(\va) \Bigg) \,-\,
\frac{\beta^2}{2}\,\Big\langle Q_N \Big\rangle_{N,t} \;.
\ee
Now by definition \eqref{gammap} of $\theta_p(\va)$, we may rewrite
\be \label{square}
\sum_{p=1}^K \big(\lambda_p^{(N)}\big)^{2}\, \theta_p(\va)\; \qp^2 \,=\,
\sum_{p=1}^{K-1} \bigg(\lambda_p^{(N)}\,\sqrt{a_p}\;\qp\bigg)^{2} \,+\, \sum_{p=1}^{K-1} \bigg(\lambda_{p+1}^{(N)}\,\frac{1}{\sqrt{a_p}}\,\qpp \bigg)^{\!2}
\ee
and plugging \eqref{square} into \eqref{Qn}, we find out that
\be \label{signQn}
Q_N \,=\, -\sum_{p=1}^{K-1} \bigg(\lambda_p^{(N)}\,\sqrt{a_p}\;\qp \,-\, \lambda_{p+1}^{(N)}\,\frac{1}{\sqrt{a_p}}\,\qpp \bigg)^{\!2} \,\leq 0 \;.
\ee
The thesis follows immediately from \eqref{t0}, \eqref{t1}, \eqref{deriv_ann} and \eqref{signQn}.
\end{proof}

\section{The annealed region of the DBM} \label{sec:ann}

In this Section we identify a region in which the quenched and the annealed pressure of the DBM coincide. The boundary delimiting this region will be given in Proposition \ref{explicit}.
\newline
Let $p^{SK}(\beta)$ be the limiting quenched pressure of an SK model at inverse temperature $\beta$ and let $p^{A}(\beta)$ be its annealed expression. By Jensen inequality:
\be\label{anndis}
p^{SK}(\beta) \,\leq\,  p^{A}(\beta) = \log 2 +\frac{\beta^2}{2}
\ee
and equality is achieved in the so called annealed region of the SK model \cite{ALR,CG,Panchenko-Book,Tala}:
\be\label{anneq}
p^{SK}(\beta) \,=\,  p^{A}(\beta) \quad \textrm{if }\beta^2\leq\frac{1}{2} \; ,
\ee
notice that the region, due to a different parametrisation, is different than the one appearing in \cite{ALR}.
This observation combined with Theorem \ref{maint} entail our result on the annealed region of the DBM.
Consider the following set of parameters:
\be\label{AK}
A_K \,\equiv\, \Big\{(\beta,\vlambda):\ \beta^2\,\lambda_p\,\theta_p(\va) \,\leq\,\frac{1}{2}\, \textrm{ for all } p=1,\dots,K \textrm{ and some } \va\in(0,\infty)^{K-1} \Big\}
\ee
where $\theta_p(\va)$ is defined in \eqref{gammap}.
\newline

\begin{theorem}\label{annt}
For $(\beta,\vlambda)\in A_K$, the quenched and the annealed pressure of the DBM coincide in the thermodynamic limit. Precisely, there exists
\be\label{annDBM}
\lim_{N\to\infty} p^{DBM}_{\Lambda_N}(\beta) \,=\, \lim_{N\to\infty}\frac{1}{N}\log\E\,Z_{\Lambda_N}(\beta) \,=\,
\log 2 + \beta^2\,\sum_{p=1}^{K-1}\lambda_p\lambda_{p+1} \;.
\ee
\end{theorem}

\begin{proof}
The lower bound \eqref{main2} found in Theorem \ref{maint} rewrites as follows:
\be\label{main 3}
\begin{split}
\liminf_{N\to\infty}p^{DBM}_{\Lambda_N}(\beta) \,\geq\,& \,\sup_{\va\in(0,\infty)^{K-1}}\,
\sum_{p=1}^K\, \lambda_p\,\bigg(p^{SK}\Big(\beta\sqrt{\lambda_p\,\theta_p(\va)}\;\Big) \,-\, p^{A}\Big(\beta\sqrt{\lambda_p\,\theta_p(\va)}\;\Big)\bigg) \,+\\
& \,+\, \log2 \,+\, \beta^2\sum_{p=1}^{K-1}\lambda_p\lambda_{p+1} \;.
\end{split}
\ee
Thanks to \eqref{anndis} and \eqref{anneq}, if $(\beta,\vlambda)\in A_K$ then the supremum in \eqref{main 3} vanishes and
\be\label{queann}
\liminf_{N\to\infty}p^{DBM}_{\Lambda_N}(\beta) \,\geq\, \log 2\,+\,\beta^2 \sum_{p=1}^{K-1}\lambda_p\lambda_{p+1} \;.
\ee
The reversed bound for $\limsup_{N\to\infty}p^{DBM}_{\Lambda_N}(\beta)$ follows immediately by Jensen inequality.
\end{proof}

Theorem \ref{annt} can be used to obtain a sufficient  condition on $(\beta,\vlambda)$  in order to have equality between quenched and annealed  pressures of the DBM.
In the following we focus on networks made up with two, three or four layers ($K\leq 4$).
\newline

\begin{proposition}\label{explicit}
Consider a $DBM$ with $K=2,3,4$ layers.
The annealed region $A_K$ defined in \eqref{AK} rewrites as
\be\label{AKexplicit}
A_K \,=\, \big\{ (\beta,\vlambda):\, 4\beta^4 \leq \phi_K(\vlambda) \big\} \;,
\ee
where we set
\begin{align} \label{phi2}
& \phi_2(\vlambda) \equiv\, \frac{1}{\lambda_1\lambda_2} \\ \label{phi3}
& \phi_3(\vlambda) \equiv\, \frac{1}{\lambda_1\lambda_2+\lambda_2\lambda_3} \\ \label{phi4}
& \phi_4(\vlambda) \equiv\, \min\{t>0:\, 1-t\,(\lambda_1\lambda_2+\lambda_2\lambda_3+\lambda_3\lambda_4)+t^2\,\lambda_1\lambda_2\lambda_3\lambda_4=0 \} \;.
\end{align}
\end{proposition}

\begin{proof}
For $K=2$, $(\beta,\lambda_1,\lambda_2)\in A_2$ if and only if
\be\label{dcon2}
\exists\, a_1>0\ \textrm{ s.t. }
\begin{cases}
a_1 \,\leq\, \dfrac{1}{2\beta^2\lambda_1} \\
\dfrac{1}{a_1} \,\leq\, \dfrac{1}{2\beta^2\lambda_2}
\end{cases}
\quad \Leftrightarrow\quad
4\beta^4\, \lambda_1\lambda_2 \,\leq\, 1 \;.
\ee
As expected we have re-obtained the same result achieved in \cite{bipartiti} by a second moment argument.

In order to extend the computations to $K>2$, we set $\theta(x,y)\equiv \dfrac{1}{x} + y$ for every $x,y>0$. The following (trivial) observation about the monotonicity of $\theta(x,y)$ will be useful:
\be\label{incdec}
\theta(x,y)\;\searrow\;\textrm{w.r.t.}\  x>0 \quad\textrm{and}\quad \theta(x,y)\;\nearrow\;\textrm{w.r.t.}\ y>0 \,.
\ee
Now for $K=3$, $(\beta,\lambda_1,\lambda_2,\lambda_3)\in A_3$ if and only if
\be\label{dcon3}
\exists\, a_1,a_2>0\ \textrm{ s.t. }
\begin{cases}
a_1 \,\leq\, \dfrac{1}{2\beta^2\lambda_1} \\
\theta(a_1,a_2) \,\leq\, \dfrac{1}{2\beta^2\lambda_2} \\
\dfrac{1}{a_2} \,\leq\, \dfrac{1}{2\beta^2\lambda_3}
\end{cases} \;.
\ee
By \eqref{incdec} one can choose without loss of generality $a_1 = \dfrac{1}{2\beta^2\lambda_1}$ and
$a_2 = 2\beta^2\lambda_3\,$. Precisely equation \eqref{dcon3} holds if and only if
\be
\theta\bigg(\frac{1}{2\beta^2\lambda_1}\,,\,2\beta^2\lambda_3\bigg) \,\leq\, \dfrac{1}{2\beta^2\lambda_2}
\quad \Leftrightarrow\quad
4\beta^4\, (\lambda_1\lambda_2 +\lambda_2\lambda_3) \,\leq\, 1 \;.\\
\ee

For $K=4$, $(\beta,\lambda_1,\lambda_2,\lambda_3,\lambda_4)\in A_4$ if and only if
\be\label{dcon4}
\exists\, a_1,a_2,a_3>0\ \textrm{ s.t. }
\begin{cases}
a_1 \,\leq\, \dfrac{1}{2\beta^2\lambda_1} \\
\theta(a_1,a_2) \,\leq\, \dfrac{1}{2\beta^2\lambda_2} \\
\theta(a_2,a_3) \,\leq\, \dfrac{1}{2\beta^2\lambda_3} \\
\dfrac{1}{a_3} \,\leq\, \dfrac{1}{2\beta^2\lambda_4}
\end{cases}\;.
\ee
Using property \eqref{incdec} of the function $\theta$, \eqref{dcon4} rewrites as:
\be
\exists\, a_2>0\ \textrm{ s.t. }
\begin{cases}
\theta\bigg(\dfrac{1}{2\beta^2\lambda_1}\,,\,a_2\bigg) \,\leq\, \dfrac{1}{2\beta^2\lambda_2} \\
\theta(a_2\,,\,2\beta^2\lambda_4) \,\leq\, \dfrac{1}{2\beta^2\lambda_3}
\end{cases} \;,
\ee
which is equivalent to
\be\label{dcon4b}
\begin{cases}
(1-4\beta^2\lambda_1\lambda_2)\,(1-4\beta^2\lambda_3\lambda_4) \,\geq\, 4\beta^4\,\lambda_2\lambda_3 \\
1 - 4\beta^4\, \lambda_1\lambda_2 \,\geq\, 0
\end{cases} \;.
\ee
Setting $t\equiv 4\beta^4$, the first inequality in \eqref{dcon4b} rewrites as $t\leq t_-\lor\,t\geq t_+$ where $t_{\pm}$ are the solutions of equation $1-t\,(\lambda_1\lambda_2+\lambda_2\lambda_3+\lambda_3\lambda_4)+t^2\,\lambda_1\lambda_2\lambda_3\lambda_4=0\,$.
Now, since $\sum_{p=1}^4 \lambda_p=1$, it is possible to prove that $t_-\leq\frac{1}{\lambda_1\lambda_2}\leq t_+\,$.
Therefore \eqref{dcon4b} is equivalent to $t\leq t_-\,$.
\end{proof}

We are interested in the $\vlambda$'s that make the region $A_K$ as small as possible.
By Proposition \ref{explicit}, we simply have to compute the infimum of $\phi_K(\vlambda)$, constraining over $\sum_{p=1}^K\lambda_p=1$ and $\lambda_p\geq0$ for every $p=1,\dots,K$.
Standard computations lead to the following

\begin{corollary}
For $K=2,3,4$
\be\label{infphi}
\inf \phi_K(\vlambda) = 4 \;.
\ee
In particular when $\beta\leq1$ the DBM is in the annealed regime for any choice of $\vlambda\,$.
Moreover the infimum of $\phi_K(\lambda)$ is reached for
\be\label{infshape}
\begin{cases}
\lambda_1=\lambda_2=\frac{1}{2} &\textrm{if }K=2\\
\lambda_2=\frac{1}{2},\, \lambda_1+\lambda_3=\frac{1}{2} &\textrm{if }K=3\\
(\lambda_4=0,\, \lambda_2=\frac{1}{2}, \lambda_1+\lambda_3=\frac{1}{2})\ \textrm{or}\ (\lambda_1=0,\, \lambda_3=\frac{1}{2}, \lambda_2+\lambda_4=\frac{1}{2}) & \textrm{if }K=4
\end{cases} \;.
\ee
\end{corollary}

These result have been obtained trough standard analytic computations for the cases $K=2,3$ and with the support of Mathematica for $K=4$. For the general $K$ case instead a refinement of the techniques is needed and could be a topic for future investigations.

These values of $\vlambda$ can be viewed as the shape that a DBM should have in order the maximally compress the annealed region. The duality among disorder-to-order transition in statistical mechanics of disordered systems and detectability-undetectability transition in machine learning  (see e.g. \cite{dualità,Peter1,Peter2,Barbier,cocco,Aurelienne,Mezard}) suggests that the knowledge of the optimal shapes stemmed from the former could play some role in the latter.

\section{A replica symmetric approximation for the DBM} \label{Sec:RS}

In this section we derive a replica symmetric expression for the intensive pressure of the DBM and we show that it is consistent with the results found in the previous section.
By Theorem \ref{annt} and Proposition \ref{explicit} an annealed region $A_K$ has been identified, even if in principle quenched and annealed pressure could coincide on a larger region of the parameters $(\beta,\lambda)$. However, in Proposition \ref{stabilityprop} we will see that the annealed solution is stable for the replica symmetric functional only in the interior of the region $A_K$. This fact suggests that $A_K$ could actually identify the whole annealed region of the DBM.

For every $p=1,\dots,K$, we consider in this section also a random external field $h_i^{(p)}$ acting on the spin $\sigma_i$ for $i\in L_p$  (see Remark \ref{addfield}). The $h_i^{(p)}$ for $i\in L_p$ are i.i.d. copies of a random variable $h^{(p)}$ satisfying $\E |h^{(p)}|<\infty$. All the $(h_i^{(p)})_{i\in L_p}$ for $p=1\dots,K$ are independent and independent also of the disorder of the process $H_{\Lambda_N}$.
We denote by $h$ the relevant parameters coming from all the above random variables.
The quenched pressure density of the model is thus
\be\label{qpres2}
p^{DBM}_{\Lambda_N}(\beta,h) \,\equiv\, \frac{1}{N}\; \E\,\log\, \sum_{\sigma}\, \exp\bigg(\!-\beta H_{\Lambda_N}(\sigma) \,+\, \sum_{p=1}^K\, \sum_{i\in L_p} h_i^{(p)}\sigma_i \,\bigg)
\ee
where $H_{\Lambda_N}$ was defined in \eqref{eq:H}.

For $y=(y_p)_{p=1,\dots,K}\in[0,\infty)^K$ the \textit{replica symmetric functional} of the DBM is defined as
\be\label{PRS}
\begin{split}
\mathcal{P}^{RS}_{\Lambda_N}(y,\beta,h) \,\equiv\; &
\sum_{p=1}^{K} \lambda_p^{(N)}\, \E_{z,h}\log\cosh\Big(\beta\,\sqrt{2}\,\sqrt{\lambda_{p-1}^{(N)}\,y_{p-1}+\lambda_{p+1}^{(N)}\,y_{p+1}}\;z + h^{(p)}\Big) \;+\\
& +\, \beta^2 \sum_{p=1}^{K-1}\lambda_p^{(N)}\lambda_{p+1}^{(N)}\,(1-y_p)\,(1-y_{p+1}) \,+\, \log2
\end{split}
\ee
where $z$ is a standard Gaussian random variable independent of $h^{(1)},\dots,h^{(p)}$, and for convenience we set  $y_0\equiv y_{K+1}\equiv \lambda_0^{(N)}\equiv \lambda_{K+1}^{(N)}\equiv 0\,$.
Its limit as $N\to\infty$ is denoted by $\mathcal{P}^{RS}(y,\beta,h,\lambda)\,$.
Definition \eqref{PRS} is motivated by the following
\newline

\begin{proposition}
For every $y=(y_p)_{p=1,\dots,K}\in[0,\infty)^K\,$, the following identity holds:
\be\label{ftc}
p^{DBM}_{\Lambda_N}(\beta,h) \,=\, \mathcal{P}^{RS}_{\Lambda_N}(y,\beta,h) \,-\, \beta^2\, \int_0^1 \Big\langle {\tilde Q}_N \Big\rangle_{N,t} \,dt \;,
\ee
where $\langle\,\cdot\,\rangle_{N,t}$ denotes the quenched Gibbs expectation associated to a suitable Hamiltonian and for every $\sigma.\tau\in\{-1,1\}^{\Lambda_N}$
\be \label{Qn_rs}
{\tilde Q}_N(\sigma,\tau) \,\equiv\;
\sum_{p=1}^{K-1} \lambda_p^{(N)}\lambda_{p+1}^{(N)}\;\big(\qp(\sigma,\tau)-y_p\big)\,\big(\qpp(\sigma,\tau)-y_{p+1}\big) \;.
\ee
\end{proposition}

\begin{proof}
For every $p=1,\dots, K$ we consider a one-body model over the $N_p$ spin variables indexed by the layer $L_p$, at inverse temperature $\beta\,\big(\lambda_{p-1}^{(N)}\,y_{p-1}+\lambda_{p+1}^{(N)}\,y_{p+1}\big)\,$ and random external field distributed as $h_p$.
For $\sigma\in\{-1,1\}^N$ and $t\in[0,1]$ we define an interpolating Hamiltonian as follows:
\be \label{Hinter_rs}
\mathcal H_N(\sigma,t) \,\equiv\,
\sqrt{t}\; H_{\Lambda_N}(\sigma) \,+\,
\sum_{p=1}^K \sum_{i\in L_p} \bigg( \sqrt{1-t}\; \sqrt{2}\,\sqrt{\lambda_{p-1}^{(N)}\,y_{p-1} \,+\, \lambda_{p+1}^{(N)}\,y_{p+1}}\;z_i^{(p)}\,+\,h^{(p)}_i \bigg)\,\sigma_i
\ee
where $z_i^{(p)}$, $i\in L_p$, $p=1,\dots, K$ are independent standard Gaussian random variables, independent also of $H_{\Lambda_N}$ defined in \eqref{eq:H}.
For $t\in(0,1)$ we introduce the interpolating pressure $\varphi_N(t)$ as
\be
\varphi_N(t) \,\equiv\, \frac{1}{N}\,\E\,\log\,\sum_{\sigma} \exp\big(\! -\beta\, \mathcal H_{N}(\sigma,t) \,\big) \ .
\ee
Observe that the quenched pressure of the DBM and a convex combination of quenched pressures of one-body models are recovered at the endpoints of $[0,1]\,$:
\begin{align} \label{t1_rs}
& \varphi_{N}(1) \,=\, p^{DBM}_{\Lambda_N}(\beta,h) \;,\\
\label{t0_rs}
& \varphi_{N}(0) \,=\,
\log 2 \,+\, \sum_{p=1}^K \lambda_p^{(N)}\; \E_{z,h}\log\cosh\bigg(\beta\,\sqrt{2}\,\sqrt{\lambda_{p-1}^{(N)}\,y_{p-1}+\lambda_{p+1}^{(N)}\,y_{p+1}}\;z+h_p\bigg) \;.
\end{align}
Gaussian integration by parts leads to the following result:
\be\label{deriv_rs}
\frac{d\phi_N(t)}{dt} \,=\, \beta^2\, \sum_{p=1}^{K-1}\lambda_p^{(N)}\lambda_{p+1}^{(N)}\;(1-y_p)\,(1-y_{p+1})
\,-\, \beta^2\, \Big\langle Q_N \Big\rangle_{N,t}
\ee
where ${\tilde Q}_N={\tilde Q}_N(\sigma,\tau)$ has been defined in \eqref{Qn_rs} and $\langle\,\cdot\,\rangle_{N,t}$ denotes the quenched Gibbs expectation associated to the Hamiltonian $\mathcal H_N(\sigma,t)+\mathcal H_N(\tau,t)$.

Therefore \eqref{ftc} follows by \eqref{t1_rs}, \eqref{t0_rs}, \eqref{deriv_rs}.
\end{proof}

\begin{remark}
Informally we say that the DBM is in the replica symmetric regime when there exists a stationary point $y^*$ of $\mathcal{P}^{RS}(y)$ such that $\int_0^1 \langle {\tilde Q}_N \rangle_{N,t}\,dt\,$ vanishes in the thermodynamic limit.
Unfortunately due to the lack of convexity in the structure of the remainder ${\tilde Q}_N$ it is not immediate to see what should be the right extremization procedure for $\mathcal{P}^{RS}(y)$.
\end{remark}

Stationary points of $\mathcal P^{RS}(y)$ satisfy the following system of self-consistent equations:
\be\label{PRScon}
y_p \,=\, \E_z\,\tanh^2 \Big(\beta\,\sqrt{2}\,\sqrt{\lambda_{p-1}y_{p-1}+\lambda_{p+1}y_{p+1}}\;z \,+\, h_p\Big)
\qquad \forall\,p=1,\dots,K\; .
\ee
From now on we assume zero external field, namely $h_p\equiv0$ for every $p=1,\dots,K\,$.
Observe that $y=0$ is a solution of \eqref{PRScon} and at this stationary point the replica symmetric functional equals the annealed pressure of the DBM (already computed in the r.h.s. of \eqref{annDBM}):
\be
\mathcal{P}^{RS}(y=0,\beta,h=0,\lambda) \,=\, \log 2 \,+\, \beta^2 \sum_{p=1}^{K-1} \lambda_p\lambda_{p+1} \;.
\ee
We are interested in the conditions on $\beta,\lambda$ that make the \textit{annealed solution} $y=0$ a stable solution of the fixed point equation \eqref{PRScon}.
It is convenient to write \eqref{PRScon} as $y=F(y)\,$, where the function $F:\R^K\to\R^K$, $F=(F_p)_{p=1,\dots,K}$ is defined by
\be
F_p(y) \,\equiv\, \E_z\tanh^2\Big(\beta\,\sqrt{2}\,\sqrt{\lambda_{p-1}y_{p-1}+\lambda_{p+1}y_{p+1}}\;z \Big) \;.
\ee
Let $\mathcal{J}_F(y) \equiv \left(\frac{\partial F_p}{\partial y_{p'}}\right)_{p,p'=1\ldots K}$ be the Jacobian matrix of $F$ at point $y\,$.
$y=0$ is a stable solution of \eqref{PRScon} if the spectral radius $\rho(\mathcal{J}_F(0))<1$, namely if all the eigenvalues of $\mathcal{J}_F(0)$ have absolute value smaller than $1$.
Gaussian integration by parts allows to compute the Jacobian matrix at $y=0$:
\be
\dfrac{\partial F_p}{\partial y_{p'}}\Big|_{y=0} \,=\,
2\,\beta^2\, \lambda_p\, (\delta_{p-1,p'}+\delta_{p+1,p'}) \;,
\ee
we denote its characteristic polynomial by
\be
\Delta_K(x) \,\equiv\, \det \big( x I - \mathcal{J}_F(0) \big) \;.
\ee
Now we confine our investigation to the cases $K=2,3,4$, as in Section \ref{sec:ann}. We have:
\begin{align}
&\Delta_2(x) \,=\, x^2 - 4\beta^4\, \lambda_1\lambda_2 \;,\\
&\Delta_3(x) \,=\, x^3 - 4\beta^4\, x\,(\lambda_1\lambda_2+\lambda_2\lambda_3) \;,\\
&\Delta_4(x) \,=\, x^4 - 4\beta^4\, x^2\,(\lambda_1\lambda_2+\lambda_2\lambda_3+\lambda_3\lambda_4) + 16\beta^8\, \lambda_1\lambda_2\lambda_3\lambda_4 \;.
\end{align}

Standard computations show the following
\newline

\begin{proposition} \label{stabilityprop}
Consider a DBM with $K=2,3,4$ layers and assume $h=0$.
The region of parameters $(\beta,\lambda)$ such that the annealed solution $y=0$ is a stable solution of the replica symmetric consistency equations \eqref{PRScon}
coincide with the interior of the region $A_K$ introduced in Section \ref{sec:ann}. Precisely:
\be\label{stabilty}
\rho (\mathcal J_F(0))<1\ \Leftrightarrow\ 4\beta^4<\phi_K(\lambda) \;,
\ee
where $\phi_K(\lambda)$ is defined by \eqref{phi2},\eqref{phi3},\eqref{phi4}.
\end{proposition}

\section{Conclusions}

While much theoretical work on the processes of learning and retrieving information in shallow neural network has been produced along the past decades, deep neural networks still escape this formalization. As the analysis of archetypal -despite quite atypical- models always played as a useful rudimentary guide, in a quest for a comprehension of neural networks, the random-weight theory (i.e. the natural setting for the statistical mechanics of disordered systems) has provided to be fundamental since the celebrated AGS theory.

With this this perspective in mind in this paper we studied, through the statistical mechanics of disordered systems, the properties of the quenched free energy of a Deep Boltzmann Machine (DBM).
The control (tunable) parameters for this model are the inverse temperature $\beta$ and the collection of the form factors $\lambda$ (i.e. the relative ratios among adjacent layers) while the order parameters are the overlaps within each layer.
We identified, in the control parameters space, a region where the quenched pressure density in the thermodynamic limit coincides with its annealed expression. A side result is the existence of the infinite volume limit for the pressure in the parameters regions that we have identified.

Inspired by the connection between the disorder-to-order transition in statistical mechanics of disordered systems and the detectability-undetectability transition in machine learning, we confined the annealed region in a space as narrow as possible. Such condition of extremality results in constraints relating noise and form factors: a collection of optimal lambdas and, remarkably, for $K=4$, the need for a small extremal layer (i.e. the size of last layer has to grow sub-linearly with respect to the total network size). We speculate this condition to be somehow expected and welcomed since learning tasks typically require information compressing. We plan to analyse networks of arbitrary depth in future works.

\begin{acknowledgements}
The authors are grateful to Francesco Guerra, Alina S\^irbu, Daniele Tantari for useful conversations.
A.B. was partially supported by MIUR via
Rete Match - Progetto Pythagoras (CUP:J48C17000250006) and by INFN and Unisalento.
P.C. was partially supported by PRIN project Statistical Mechanics and Complexity (2015K7KK8L). D.A. and E.M.
were partially supported by Progetto Almaidea 2018,

\end{acknowledgements}


\begin{thebibliography}{99}

%\bibitem{Agliari-PRL1}  E. Agliari, et al., {\em Multitasking associative networks}, Phys. Rev. Lett. \textbf{109}, 268101, (2012).

%\bibitem{Agliari-Barattolo} E. Agliari, et al., {\em Neural Networks retrieving binary patterns in a sea of real ones}, J. Stat. Phys. \textbf{168}, 1085, (2017).

%\bibitem{ABT} E. Agliari, A. Barra, B. Tirozzi, {\em  Free energies of Boltzmann Machines: self-averaging, annealed and replica symmetric approximations in the thermodynamic limit}, Journal of Statistical Mechanics: Theory and Experiment, 033301 (2019)

\bibitem{AMT2018} E. Agliari, D. Migliozzi, D. Tantari, {\em Non-convex multi-species Hopfield models}, J. Stat. Phys. \textbf{172}(5), 1247-1269, (2018).

\bibitem{ALR} M. Aizenman, J.L. Lebowitz, D. Ruelle, {\em Some Rigorous Results on the Sherrington-Kirkpatrick Spin Glass Model}, Communications in Mathematical Physics 112, 3-20 (1987)

\bibitem{Amit} D. J. Amit, {\em Modeling brain functions}, Cambridge University Press, 1989

\bibitem{AC} A. Auffinger, W. K. Chen {\em Free Energy and Complexity of Spherical Bipartite Models}. Journal of Statistical Physics, 157, 1, 40–59 (2014)




%\bibitem{angel-learning} A. Engel, C. Van den Broeck, {\em Statistical mechanics of learning}, Cambridge University Press, 2001

%\bibitem{Chiara} M. Baity-Jesi, et al., {\em Comparing dynamics: Deep neural networks versus glassy systems}, preprint arXiv:1803.06969, (2018).

%\bibitem{BarraEquivalenceRBMeAHN} A. Barra, A. Bernacchia, E. Santucci, P. Contucci, {\em On the equivalence among Hopfield neural networks and restricted Boltzman machines}, Neural Networks \textbf{34}, 1-9, (2012).

%\bibitem{Barra-RBMsPriors1} A. Barra, G. Genovese, P. Sollich, D. Tantari, {\em Phase transitions of Restricted Boltzmann Machines with generic priors}, Phys. Rev. E \textbf{96}, 042156, (2017).

\bibitem{BCMT} A. Barra, P. Contucci, E. Mingione, D. Tantari, {\em Multi-species mean field spin glasses: Rigorous results}, Annales Henri Poincar\'e 16(3), 691-708 (2015)
    
\bibitem{dualità} A. Barra, A. Bernacchia, E. Santucci, P. Contucci, {\em On the equivalence of hopfield networks and boltzmann machines}, Neur. Netws. \textbf{34}, 1-9, (2012).
%\bibitem{Barra-JSP2010} \textbf{(non è citato nel testo)} A. Barra, G. Genovese, F. Guerra, {\em The replica symmetric approximation of the analogical neural network}, Journal of Statistical Physics 140(4), 784-796 (2010)

\bibitem{bipartiti} A. Barra,  G. Genovese, F. Guerra, {\em Equilibrium statistical mechanics  of bipartite spin systems}, Journal of Physics A 44, 245002 (2011)

\bibitem{Peter1} A. Barra, G. Genovese, P. Sollich, D. Tantari,  {\em Phase transitions in Restricted Boltzmann Machines with generic priors}, Phys. Rev. E, 96(4), 042156, (2017).

\bibitem{Peter2} A. Barra, G. Genovese, P. Sollich, D. Tantari,  {\em Phase diagram of restricted Boltzmann machines and generalized Hopfield networks with arbitrary priors}, Phys. Rev. E \textbf{97}(2), 022310, (2018).

%\bibitem{Bates} E. Bates, L. Sloman, Y. Sohn, {\em Replica Symmetry Breaking in Multi-species Sherrington–Kirkpatrick Model}, J. Stat. Phys. \textbf{174}(2):333, (2019).

%\bibitem{Bovier1} A. Bovier, V. Gayrard, {\em Hopfield models as generalized random mean field models}, Mathematical aspects of spin glasses and neural networks, 3-89, Birkhauser, Boston (1998).

%\bibitem{Bovier2} A. Bovier, et al., {\em Gibbs states of the Hopfield model in the regime of perfect memory}, Prob. Theor. $\&$ Rel. Fields       \textbf{100}(3):329, (1994).

%\bibitem{Bovier3} A. Bovier, et al., {\em Gibbs states of the Hopfield model with extensively many patterns}, J. Stat. Phys. \textbf{79}(1-2):395, (1995).
\bibitem{Barbier} J. Barbier, F. Krzakala, N. Macris, L. Miolane, L. Zdeborova, {\em  Optimal errors and phase transitions in high-dimensional generalized linear models}, Proc. Natl. Acad. Sci. (USA) \textbf{116}(12), 5451-5460, (2019).
    
\bibitem{BovierBook} A. Bovier, P. Picco,  {\em Mathematical aspects of spin glasses and neural networks}, Springer Press, 2012

%\bibitem{Univ1} P. Carmona, Y. Hu, {\em Universality in Sherrington–Kirkpatrick's spin glass model}, Ann. Henri Poincar\`e \textbf{42}, 2, (2006).

%\bibitem{LeCun1} A. Choromanska, et al., {\em The loss surfaces of multilayer networks}, Artif. Int. $\&$ Statistics 192-204, (2015)

%\bibitem{GPU2} R. Collobert, J. Weston, {\em A unified architecture for natural language processing: Deep neural networks with multitask learning}, Proc. $25^{th}$ Int. Conf. on Mach. Learn. ACM, (2008).

\bibitem{cocco} S. Cocco, R. Monasson, V. Sessak, {\em High-dimensional inference with the generalized Hopfield model: Principal component analysis and corrections}, Phys. Rev. E \textbf{83}(5), 051123, (2011).

\bibitem{Coolen} A.C.C. Coolen, R. Kuhn, P. Sollich, {\em Theory of neural information processing systems}, Oxford University Press, 2005

\bibitem{CG} P. Contucci, C. Giardin\`a, {\em Perspectives on spin glasses}, Cambridge University Press, 2013

%\bibitem{Decelle} A. Decelle, F. Krzakala, C. Moore, L. Zdeborova, {\em Inference and phase transitions in the detection of modules in sparse networks}, Phys. Rev. Lett. \textbf{107}(6), 065701, (2011)	

%\bibitem{Dotsenko} V. Dotsenko, {\em An introduction to the theory of spin glasses and neural networks},  World Scientific, (1995).

%\bibitem{Dotsenko1} V. Dotsenko, et al., {\em Statistical mechanics of Hopfield-like neural networks with modified interactions}, J. Phys. A \textbf{24}, 2419, (1991).

%\bibitem{Dotsenko2} V. Dotsenko, B. Tirozzi, {\em Replica symmetry breaking in neural networks with modified pseudo-inverse interactions}, J. Phys. A \textbf{24}:5163-5180, (1991).

%\bibitem{Albert2} A. Fachechi, et al., {\em Dreaming neural networks: rigorous results}, JSTAT  083503, (2019).

%\bibitem{French} R.M. French, {\em Catastrophic forgetting in connectionist networks}, Trends in cognitive sciences \textbf{3}.4: 128-135, (1999).

%\bibitem{Gabrie} M. Gabri\`e, et al., {\em Training Restricted Boltzmann Machine via the Thouless-Anderson-Palmer free energy}, Adv. Neur. Inf. Proc. Sys., 640-648, (2015).

%\bibitem{Genovese} G. Genovese, {\em Universality in bipartite mean field spin glasses}, J. Math. Phys. \textbf{53}(12):123304, (2012).

\bibitem{Aurelienne} A. Decelle, F. Krzakala, C. Moore, L Zdeborova,  {\em Asymptotic analysis of the stochastic block model for modular networks and its algorithmic applications}, Phys. Rev, E \textbf{84}(6), 066106, (2011).

\bibitem{DLbook} I. Goodfellow, Y. Bengio, A. Courville, {\em Deep Learning}, M.I.T. Press, 2017

\bibitem{Guerra} F. Guerra, {\em Broken replica symmetry bounds in the mean field spin glass model}, Communications in Mathematica Physics 233(1), 1-12 (2003)

\bibitem{GT} F. Guerra, F.L. Toninelli, {\em The Thermodynamic Limit in Mean Field Spin Glass Models}, Communications in Mathematical Physics 230(1), 71-79 (2002)


%\bibitem{Clouds1} I.A.T. Hashem, et al.,  {\em The rise of “big data” on cloud computing: Review and open research issues}, Infor. Sys. \textbf{47}:98, (2015).

%\bibitem{Hopfield} J.J. Hopfield, {\em Neural networks and physical systems with emergent collective computational abilities},  Proceedings of the national academy of sciences 79.8 (1982): 2554-2558.

%\bibitem{Jaynes} E.T. Jaynes, {\em Information theory and statistical mechanics}, Phys. Rev. \textbf{106}.4:620, (1957).

%\bibitem{Kappen} H.J. Kappen, F.D.B. Rodriguez, {\em Efficient learning in Boltzmann machines using linear response theory}, Neur. Comp. \textbf{10}(5):1137, (1998).

%\bibitem{Kinzel} W. Kinzel, M. Opper, {\em Dynamics of learning}, in: E. Domany, J.L. van Hemmen, K. Schulten (Eds.) {\em Models of neural networks}, Springer, Berlin, 149-172 (1991).

%\bibitem{Lenka1} F. Krzakala, et al., {\em Statistical-physics-based reconstruction in compressed sensing}, Phys. Rev. X \textbf{2}(2), 021005, (2012).

%\bibitem{Lenka2} F. Krzakała, A. Montanari, F. Ricci-Tersenghi, G. Semerjian, L. Zdeborova, {\em Gibbs states and the set of solutions of random constraint satisfaction problems}, Proc. Natl. Acad. Sci. \textbf{104}:(25),10318, (2007).

%\bibitem{Krotov1} D. Krotov, J.J. Hopfield, {\em Dense associative memory for pattern recognition}, Adv. Neur. Inform. Proc. Sys. 1172–1180, (2016).

%\bibitem{Krotov2} D. Krotov, J.J. Hopfield, {\em Dense associative memory is robust to adversarial inputs}, Neur. Comp. 30.(12):3151, (2018).

%\bibitem{kirkpatrick} S. Kirkpatrick, et al., {\em Optimization by simulated annealing}, Science \textbf{220}:671-680, (1983).

%\bibitem{kohonen} T.O. Kohonen, Self-organization and Associative Memory, Springer, Berlin (1984).

%\bibitem{Dimitry} D. Krotov, J.J. Hopfield, {\em Dense associative memory is robust to adversarial inputs}, arXiv:1701.00939, (2017).

%\bibitem{DL1} Y. Le Cun, Y. Bengio, G. Hinton, {\em Deep learning}, Nature \textbf{521}:436-444, (2015).

%\bibitem{Enzo} E. Marinari, {\em Forgetting Memories and their Attractiveness}, arXiv:1805.12368, (2018).

%\bibitem{Neuro1} P. Maquet, {\em The role of sleep in learning and memory}, Science \textbf{294}.5544:1048, (2001).

%\bibitem{Neuro2} J.L. McGaugh, {\em Memory - a century of consolidation}, Science \textbf{287}.5451:248-251, (2000).

%\bibitem{Mezard} M. Mezard, {\em Mean-field message-passing equations in the Hopfield model and its generalizations}, Phys. Rev. E \textbf{95}(2), 022117, (2017).

\bibitem{Mezard} M. M\'ezard, {\em Mean-field message-passing equations in the Hopfield model and its generalizations}, Phys. Rev. E \textbf{95}(2), 022117, (2017).

\bibitem{MezardMontanari} M. M\'ezard, A. Montanari, {\em Information, Physics, and Computation}, Oxford University Press, 2009

\bibitem{MPV} M. M\'ezard, G. Parisi, M.A. Virasoro, {\em Spin glass theory and beyond: an introduction to the replica method and its applications}, World Scientific, 1987

%\bibitem{Zecchina} M. Mezard, G. Parisi, R. Zecchina, {\em Analytic and algorithmic solution of random satisfiability problems}, Science \textbf{297}.5582:812-815, (2002).

%\bibitem{Metha} P. Mehta, D.J. Schwab, {\em An exact mapping between the variational renormalization group and deep learning}, preprint, arXiv:1410.3831, (2014)

%\bibitem{Monasson} R. Monasson, R. Zecchina, S. Kirkpatrick, B. Selman, L. Troyansky, {\em Determining computational complexity from characteristic phase transitions}, Nature \textbf{400}(6740), 133, (1999).

\bibitem{Nishimori} H. Nishimori, {\em Statistical physics of spin glasses and information processing: an introduction},  Clarendon Press, 2001

\bibitem{Panchenko-Book} D. Panchenko, {\em The Sherrington-Kirkpatrick model}, Springer Press, 2013

\bibitem{PanchenkoMSK} D. Panchenko, {\em The free energy in a multi-species Sherrington–Kirkpatrick model}, The Annals of Probability 43(6), 3494-3513 (2015)

%\bibitem{Pastur} L. Pastur, M. Shcherbina, B. Tirozzi, {\em The replica-symmetric solution without replica trick for the Hopfield model}, J. Stat. Phys. \textbf{74}(5-6):1161, (1994).

\bibitem{Hinton1} R. Salakhutdinov, G. Hinton, {\em Deep Boltzmann machines}, Proceedings of the Twelth International Conference on Artificial Intelligence and Statistics, PMLR 5, 448-455 (2009)

%\bibitem{Hugo} R. Salakhutdinov, H. Larochelle, {\em Efficient learning of deep Boltzmann machines}, Proc. thirteenth int. conf. on artificial intelligence and statistics, 693, 2010.

%\bibitem{GPU1} J. Sanders, E. Kandrot, {\em CUDA by example: an introduction to general-purpose GPU programming}, Addison-Wesley Professional, (2010).

%\bibitem{Bialek} E. Schneidman, M.J. Berry II, R. Segev, M. Bialek, {\em Weak pairwise correlations imply strongly correlated network states in a neural population} Nature  \textbf{440}(7087):1007, (2006).

%\bibitem{sompo-learning} H.S. Seung, et al., {\em Statistical mechanics of learning from examples}, Phys. Rev. A \textbf{45}(8):6056, (1992).

%\bibitem{Tre} N. Srivastava, R. Salakhutdinov, {\em Multimodal learning with deep boltzmann machines}, Adv. Neural Inform. Proc. Sys. , 2222, (2012).

%\bibitem{Tirozzi} L. Pastur, et al., {\em On the replica symmetric equations for the Hopfield model}, J. Math. Phys. \textbf{40}(8): 3930, (1999).

%\bibitem{Pastur} L. Pastur, et al., {\em The replica-symmetric solution without replica trick for the Hopfield model}, J. Stat. Phys. \textbf{74}(5-6):1161, (1994).

\bibitem{Tala} M. Talagrand, {\em Spin glasses: a challenge for mathematicians. Cavity and mean field models}, Springer, 2003

%\bibitem{Tala1} M. Talagrand, {\em Rigorous results for the Hopfield model with many patterns}, Prob. Theor. $\&$ Rel. Fiel. \textbf{110}(2):177, (1998).

%\bibitem{Tala2} M. Talagrand, {\em Exponential inequalities and convergence of moments in the replica-symmetric regime of the Hopfield model}, Ann. Prob. 1393-1469, (2000).

%\bibitem{Monasson2} J. Tubiana, R. Monasson, {\em Emergence of Compositional Representations in Restricted Boltzmann Machines}, Phys. Rev. Lett. \textbf{118}.13:138301, (2017).

%\bibitem{Cloud2} L. Zhou, et al., {\em Machine learning on big data: Opportunities and challenges}, Neurocomputing \textbf{237}:350, (2017).

\end{thebibliography}
\end{document}